\documentclass[11pt]{article}
\usepackage[english]{babel}

\usepackage{amsmath,amssymb,amsthm}

\usepackage[pdfpagemode=UseOutlines,
bookmarks=true, bookmarksopenlevel=3, bookmarksopen=true,
bookmarksnumbered=true, unicode=true,
pdffitwindow=true,pdfpagelayout=SinglePage,pdfhighlight=/P,
colorlinks=true,linkcolor=blue,citecolor=blue,urlcolor=blue]{hyperref}

\advance\textwidth 20mm  \advance\oddsidemargin -10mm
\advance\textheight 30mm \advance\topmargin -18mm

\allowdisplaybreaks[4]

\theoremstyle{plain}
\newtheorem{theorem}{Theorem}
\newtheorem{statement}[theorem]{Statement}
\newtheorem{lemma}[theorem]{Lemma}
\theoremstyle{remark}
\newtheorem{remark}{Remark}

\newcommand{\Integer}{\mathbb{Z}}
\newcommand{\const}{\mathop{\rm const}}

\title{Integrable M\"obius invariant evolutionary lattices\\ of second order}
\author{V.E. Adler\footnote{L.D. Landau Institute for Theoretical Physics,
Ac.~Semenov str.~1-A, 142432 Chernogolovka, Russian Federation. E-mail:
adler@itp.ac.ru}}
\date{April 29, 2016}

\begin{document}\thispagestyle{empty}
\maketitle

\begin{abstract}
We solve the classification problem for integrable lattices of the form
$u_{,t}=f(u_{-2},\dots,u_2)$ under the additional assumption of invariance with
respect to the group of linear-fractional transformations. The obtained list
contains 5 equations, including 3 new. Difference Miura type substitutions are
found which relate these equations with known polynomial lattices. We also
present some classification results for the generic lattices.
\medskip

\end{abstract}

\section{Introduction}\label{s:intro}

The known integrable differential-difference equations of the form
\begin{equation}\label{utm}
 u_{,t}=f(u_{-m},\dots,u_m),\quad u_j=u(t,n+j)
\end{equation}
include, first of all, the Bogoyavlensky lattices \cite{Narita_1982, Itoh_1987,
Bogoyavlensky_1988, Bogoyavlensky_1991} and their modifications related by
Miura type substitutions \cite{Suris_2003, Svinin_2011, Berkeley_Igonin_2015}.
More general families of the lattices were considered in
\cite{Hu_Clarkson_Bullough, Adler_Postnikov_2011, Garifullin_Yamilov_2012,
Garifullin_Mikhailov_Yamilov_2014}, relations with other discrete models were
studied in \cite{Papageorgiou_Nijhoff_1996, Nijhoff_1996, Suris_2003,
Levi_Petrera_Scimiterna_2007, Adler_Postnikov_2014}. The classification of
integrable equations (\ref{utm}) at $m=1$ was obtained by Yamilov
\cite{Yamilov_1983, Yamilov_1984, Yamilov_2006}. However, already the case
$m=2$ turns out to be essentially more complicated and it remains an open
problem till now. The goal of this paper is to obtain some preliminary results
for $m=2$, including the classification of the lattices which are invariant
with respect to the linear-fractional transformations of the variable $u$.

Let us remind that the symmetry approach is a most effective tool for solving
of classification problems \cite{Mikhailov_Shabat_Yamilov_1987,
Levi_Yamilov_1997}. It is based on the fact that integrable equations admit
generalized symmetries and conservation laws of arbitrarily high order.
Adopting this property as a definition, one can build an infinite sequence of
relations which must hold for the right hand side of any integrable lattice
(\ref{utm}). Solving of the classification problem amounts to the analysis of
this over-determined system of differential-functional equations. Moreover, one
can expect that several first relations give not only the necessary, but also
the sufficient integrability conditions. For instance, in the case $m=1$ it is
sufficient, according to Yamilov \cite{Yamilov_1983}, to consider just 3
conditions, namely
\begin{equation}\label{1.cond}
 D_t(\log f_1)=(T-1)(\sigma),\quad T(rf_{-1})+rf_1=0,\quad
 D_t(\log r)+2f_0=(T-1)(\mu).
\end{equation}
Here and further on, $T:u_j\to u_{j+1}$ is the shift operator, $D_t$ is the
evolutionary differentiation in virtue of the lattice equation, the subscript
$j$ denote the partial derivative with respect to the variable $u_j$. The
integrability conditions mean that the function $f$ should be such that
equations (\ref{1.cond}) be solvable with respect to the unknown functions
$\sigma,r,\mu$ of the variables $u_j$.

The derivation and testing of integrability conditions for a {\em concrete}
equation are, in principle, not difficult for any $m$ \cite{Adler_2014,
Adler_2016}. However, the classification requires the analysis of these
conditions for an {\em undetermined} function $f$, and this problem turns out
to be extremely difficult. Before solving it in the general setting, it makes
sense to consider special cases under one or another simplifying assumption. In
this paper, the role of such additional condition is played by the invariance
of the lattice under the group of M\"obius transformations $u_j=\frac{\alpha
u_j+\beta}{\gamma u_j+\delta}$. This simplifies the problem drastically.
Indeed, one can easily see that the classification problem completely
disappears at $m=1$, because there is just one such lattice,
\begin{equation}\label{i.SV}
 u_{,t}=Y=\frac{(u_1-u)(u-u_{-1})}{u_1-u_{-1}}.
\end{equation}
This equation turns out to be integrable; a simplest way to demonstrate this is
to rewrite it in terms of the cross-ratios which satisfy the Volterra lattice:
\[
 X_{,t}=X(X_1-X_{-1}),\qquad
 X=\frac{(u_1-u)(u_{-1}-u_{-2})}{(u_1-u_{-1})(u-u_{-2})}.
\]
This example is well known and equation (\ref{i.SV}) is often called the
Schwarzian Volterra lattice, by analogy with the Schwarzian KdV equation in the
continuous case. At $m=2$, the general form of M\"obius invariant lattices
reads
\begin{equation}\label{i.ut}
 u_{,t}=YF(X,T(X)),
\end{equation}
that is, the unknown function of 5 variables is replaced with a function of
just 2 variables. The analysis of the integrability conditions in this case is
not quite trivial, but comprehensible.

The outline of the paper is the following. Section \ref{s:cond} contains first
6 integrability conditions for the lattices (\ref{utm}) at $m=2$, as well as
few simplest corollaries from these conditions. By use of these relations, we
found the general form of integrable lattices in section \ref{s:->3}, with the
right hand side represented through functions of 3 variables. Further
considerations require much efforts and we postpone them for the future work.
Section \ref{s:class} contains the main result of the paper, the classification
of equations of the special form (\ref{i.ut}). The answer is given by a list of
5 equations (Theorem \ref{th:list}). One of them is, as expected, the higher
symmetry of the lattice (\ref{i.SV}), another one is the Schwarzian version of
the Bogoyavlensky lattice \cite{Papageorgiou_Nijhoff_1996, Nijhoff_1996}. The
rest equations were, probably, not studied before. In section \ref{s:subst} we
discuss the Miura type substitutions which link these equations with the
polynomial lattices mentioned above, from the papers
\cite{Adler_Postnikov_2011, Garifullin_Yamilov_2012}. The concluding section
\ref{s:generalizations} contains some generalizations of these examples.

\section{Necessary integrability conditions}\label{s:cond}

Let the lattice (\ref{utm}) be integrable, that is, let it admit generalized
symmetries and conservation laws of order arbitrarily large. This makes
possible to prove the solvability of equations
\begin{equation}\label{fRG}
 D_t(G)=[f_*,G],\qquad D_t(R)+f_*^\dag R+Rf_*=0,
\end{equation}
where $f_*=f_mT^m+f_{m-1}T^{m-1}+\dots+f_{-m}T^{-m}$ is the linearization
operator of equation (\ref{utm}), $G$ and $R$ are formal Laurent series with
respect to the powers of $T$ or $T^{-1}$, with the coefficients depending on
the dynamical variables $u_j$ \cite{Yamilov_2006,
Mikhailov_Shabat_Yamilov_1987, Levi_Yamilov_1997}. Moreover, several first
terms of the series $G$ may be chosen the same as in the operator $f_*$,
without loss of generality \cite{Adler_2014}. More precisely, equations
(\ref{fRG}) admit solutions of the form
\begin{gather*}
 G=f_*+\sigma+\theta T^{-1}+\omega T^{-2}+\ldots,\qquad
 \tilde G=f_*+\tilde\sigma+\tilde\theta T+\tilde\omega T^2+\ldots,\\
 R=rT^l+sT^{l-1}+\ldots~,\quad r\ne0,
\end{gather*}
with some exponent $l$ in the range $0\le l\le m-1$, and, additionally, the
relation $\tilde G^\dag R=-RG$ holds, where $(aT^j)^\dag=T^{-j}a$. Solvability
of equations (\ref{fRG}) with respect to the coefficients of these series
defines the necessary integrability conditions.

Further on, we consider only the case $m=2$,
\begin{equation}\label{ut2}
 u_{,t}=f(u_{-2},u_{-1},u,u_1,u_2),
\end{equation}
and assume that the following non-degeneracy conditions are fulfilled:
$f_{-2}\ne0$, $f_2\ne0$, as well as $f_{-1}\ne0$ or $f_1\ne0$ (if
$f_1=f_{-1}=0$ then the equation splits into the pair of first order equations
on the odd and even sublattices). An easy computation proves that equations
(\ref{fRG}) yield the following conditions, in few first orders of $T$.

\begin{statement}\label{st:cond}
If equation (\ref{ut2}) is integrable then there exist functions $\sigma$,
$\theta$, $\omega$, $r$, $s$, $\mu$ of the dynamical variables $u_j$ and an
exponent $l$ equal to 0 or 1, such that the following relations hold:
\begin{align}
\label{2.cond-1}
 & D_t(f_2)=f_2(T^2-1)(\sigma),\\
\label{2.cond-2}
 & D_t(f_1)-f_1(T-1)(\sigma)=f_2T^2(\theta)-T^{-1}(f_2)\theta,\\
\label{2.cond-3}
 & D_t(f_0+\sigma)-(T-1)(T^{-1}(f_1)\theta)
   =(T^2-1)(T^{-2}(f_2)\omega),\\
\label{2.cond-4}
 & T^2(rf_{-2})+rT^l(f_2)=0,\\
\label{2.cond-5}
 & T^2(sf_{-2})+sT^{l-1}(f_2)+T(rf_{-1})+rT^l(f_1)=0,\\
\label{2.cond-6}
 & D_t(\log r)+2f_0=(T-1)(\mu).
\end{align}
\end{statement}

\begin{remark}
If we set $f_2=f_{-2}=0$ and $l=0$ then equations (\ref{2.cond-1}),
(\ref{2.cond-4}) disappear, equations (\ref{2.cond-2}), (\ref{2.cond-5}),
(\ref{2.cond-6}) turn into conditions (\ref{1.cond}) and equation
(\ref{2.cond-3}) takes the form $D_t(f_0+\sigma)=(T-1)(\nu)$. This gives on
more condition for the first order lattices, but it turns out to be redundant
(that is, it holds automatically if the conditions (\ref{1.cond}) hold).

One can conjecture that integrability conditions
(\ref{2.cond-1})--(\ref{2.cond-6}) are not only necessary, but also sufficient,
like conditions (\ref{1.cond}) in the case $m=1$. The answer on this question
may be obtained only after the complete classification. The examples exist
which demonstrate that no one condition can be dropped out of this set. For
instance, the lattice
\[
 u_{,t}=u_2-u_{-2}+c(u_1-u_{-1})^2,\quad c\ne0
\]
satisfies all conditions except for (\ref{2.cond-3}) (moreover, conditions
(\ref{2.cond-4})--(\ref{2.cond-6}) hold both for $l=0$ and $l=1$).
\end{remark}

\begin{remark}
Thanks to the relation $\tilde G^\dag=-RGR^{-1}$, the equations for
coefficients of the series $\tilde G$ are consequences of equations for $G,R$.
In spite of this, it is sometimes convenient to take these equations into
account, in order to obtain the corollaries which are not symmetric with
respect to the shifts. To this end, it is sufficient to use the substitution
\begin{equation}\label{tilde}
 (\sigma,\theta,\omega)\to(\tilde\sigma,\tilde\theta,\tilde\omega),\quad
 \partial_i\to\partial_{-i},\quad T^i\to T^{-i}.
\end{equation}
\end{remark}

Let us write down several simplest corollaries of the integrability conditions
which will be used in the next sections. To prove the following statement one
need, in fact, only the conditions (\ref{2.cond-1}), (\ref{2.cond-2}),
(\ref{2.cond-4}) and symmetry (\ref{tilde}).

\begin{statement}\label{st:rho}
If the lattice (\ref{ut2}) is integrable then $\rho=\log f_2$ satisfies the
following relations:
\begin{align}
\label{Z11}
 & \rho_{-2,2}=0, \\
\label{Z12}
 & T^2(\rho_{-2,1})f_{-2}+\rho_{-2,1}T(f_2)=0,\\
\label{Z13}
 & T(\rho_{-1,2})f_{-2}+T^{-1}(\rho_{-1,2})T(f_2)=0.
\end{align}
\end{statement}
\begin{proof}
If $l=1$ then the relation (\ref{Z11}) follows from equation (\ref{2.cond-4})
immediately, but the case $l=0$ requires more complicated reasoning. Counting
of the variables involved in the left hand sides of (\ref{2.cond-1}) and
(\ref{2.cond-2}) proves that functions $\sigma$ and $\theta$ may depend on
$u_{-4},\dots,u_2$ only. Then, applying $\partial_{-2}\partial_4$ to
(\ref{2.cond-1}) and $\partial_{-3}\partial_3$ to (\ref{2.cond-2}) yields,
respectively,
\[
 T^2(f_2)\rho_{-2,2}=T^2(\sigma_{-4,2}),\qquad f_1T(\sigma_{-4,2})=0.
\]
This implies $f_1\rho_{-2,2}=0$. We obtain also $f_{-1}\tilde\rho_{-2,2}=0$,
where $\tilde\rho=\log f_{-2}$, by use of the symmetry (\ref{tilde}). It
follows from (\ref{2.cond-4}) that, if one of the functions $\rho_{-2,2}$ or
$\tilde\rho_{-2,2}$ vanishes, then this is true for the second one as well.
Taking into account the non-degeneracy condition $f_{-1}\ne0$ or $f_1\ne0$, we
arrive to (\ref{Z11}).

In order to prove (\ref{Z12}), (\ref{Z13}), let us partially integrate equation
(\ref{2.cond-1}), by substitution
$\sigma=\hat\sigma-\rho_{-1}T^{-1}(f)-\rho_{-2}T^{-2}(f)$. This brings to an
equivalent equation with a function $\hat\sigma$ which depends on a reduced set
of variables:
\[
 (T^2(\rho_{-2})+\rho_0)f+(T^2(\rho_{-1})+\rho_1)T(f)+\rho_2T^2(f)
 =(T^2-1)(\hat\sigma(u_{-2},\dots,u_2)).
\]
The application of $\partial_{-2}\partial_3$ and $\partial_{-1}\partial_4$
brings to the desired equations.
\end{proof}

Several other useful relations can be derived by analogous calculations, for
instance, condition (\ref{2.cond-2}) implies the equations
\begin{align}
\label{Z21}
 & T(\rho_{-2,1})f_1
    +T(\rho_{-2})f_{1,2}-T^{-1}(\tilde\rho_2)f_{-1,2}=0,\\
\label{Z22}
 & T^{-1}(\tilde\rho_{-1,2})f_{-1}
    +T^{-1}(\tilde\rho_2)f_{-2,-1}-T(\rho_{-2})f_{-2,1}=0.
\end{align}

\section{Reducing to functions of three arguments}\label{s:->3}

In this section, we resolve condition (\ref{2.cond-4}) and relations from
Statement \ref{st:rho}. This casts the lattices (\ref{ut2}) to several basic
types, with the right hand side expressed through few arbitrary functions of no
more than three arguments.

\begin{theorem}\label{th:2.types}
Any integrable lattice (\ref{ut2}) belongs to one of the following types (types
I, II correspond to the case $l=0$ and the rest ones to the case $l=1$):
\begin{align*}
 &{\rm I}  && u_{,t}=b(T(a_{-1})-T^{-1}(a_1))+c,\\
 &{\rm II} && u_{,t}=b\exp\bigl(k(u)(T(a_{-1})-T^{-1}(a_1))\bigr)+c,\\
 &{\rm III}&& u_{,t}=T^{-1}(p)a_{-1}T(b)- pT^{-1}(a)b_1+c,\\
 &{\rm IV} && u_{,t}= \frac{T^{-1}(p)a_{-1}}{a+T(b)}
                     -\frac{pb_1}{T^{-1}(a)+b}+c,\\
 &{\rm V}  && u_{,t}=\frac{a_{-1,1}}{T^{-1}(A)A}+b,\quad A=T(a_{-1})-a_1,\\
 &{\rm VI} && u_{,t}=\frac{a_{-1,1}}{T^{-1}(A)A}+b,\quad
        A=\exp\bigl(p(T(a_{-1})-a_1)\bigr)-\frac{\alpha}{p},\quad
        \alpha=\const.
\end{align*}
In all equations, $a,b,c$ denote functions of $u_{-1},u,u_1$; $p=p(u,u_1)$.
\end{theorem}

\begin{proof}
{\em Case $l=0$.} Let us denote $rf_{-2}=T^{-1}(R)$, then equation
(\ref{2.cond-4}) takes the form
\[
 T(R)f_{-2}+T^{-1}(R)f_2=0,\quad R=R(u_{-1},u,u_1).
\]
Let $R=a_{-1,1}$, then $f=F(u_{-1},u,u_1,v)$, where $v=T(a_{-1})-T^{-1}(a_1)$.
Substitution into equation $\rho_{-2,2}=0$ gives $(\log F')''=0$, where prime
denotes the derivative with respect to $v$; therefore, $F''=k(u_{-1},u,u_1)F'$.
The solutions of this equation bring to the types I and II, in the cases $k=0$
and $k\ne0$, respectively. In the case II, the fact that function $k$ does not
depend on $u_{\pm1}$ follows from conditions (\ref{Z12}), (\ref{Z13}). It is
easy to prove that these relations are equivalent to
\[
 (T-1)\left(\frac{k_1T(k_1)}{T(R)F'}\right)=0,\qquad
 (T-1)\left(\frac{T^{-1}(k_{-1})k_{-1}}{T^{-1}(R)F'}\right)=0
\]
and therefore
\[
 k_1T(k_1)=\alpha T(R)F',\qquad T^{-1}(k_{-1})k_{-1}=\beta T^{-1}(R)F',
\]
with constant $\alpha,\beta$. Differentiation of the first equation with
respect to $u_{-2}$ and the second one with respect to $u_2$ yields
$\alpha=\beta=0$, taking the equation $F''=kF'\ne0$ into account, so that
$k_{\pm1}=0$.
\smallskip

{\em Case $l=1$.} Let us denote $rf_{-2}=T^{-1}(R)$ and rewrite
(\ref{2.cond-4}) as the system
\begin{equation}\label{frR}
 f_{-2}=\frac{T^{-1}(R)}{r},\quad f_2=-\frac{R}{T^{-1}(r)},\quad
 r=r(u_{-1},u,u_1,u_2),\quad R=R(u_{-1},u,u_1,u_2).
\end{equation}
This implies $T(\rho_{-2,1})=-(\log r)_{-1,2}$, $\rho_{-1,2}=(\log R)_{-1,2}$
and the comparison of (\ref{Z12}), (\ref{Z13}) with the relation
$T(R)f_{-2}+T^{-1}(R)T(f_2)=0$ brings to equations
\begin{equation}\label{rR}
 (\log r)_{-1,2}=\alpha R,\qquad (\log R)_{-1,2}=\beta R,
\end{equation}
where $\alpha,\beta$ are constants. Moreover, cross-differentiation of
(\ref{frR}) gives
\begin{equation}\label{rRla}
 \frac{f_{-2,2}}{T^{-1}(R)R}= -\frac{r_2}{r^2R}
  = T^{-1}\left(\frac{r_{-1}}{r^2R}\right)= \lambda(u_{-1},u,u_1).
\end{equation}

First, let $\lambda=0$, then $r_{-1}=r_2=0$. Let us denote $r=-1/p(u,u_1)$,
$R=h_{-1,2}$, then integration of (\ref{frR}) gives
\[
 f=T^{-1}(p)h_{-1}-pT^{-1}(h_2)+c(u_{-1},u,u_1).
\]
Solutions of the second equation (\ref{rR}) are:
\[
 (\beta=0)\quad h=aT(b)+\tilde a+T(\tilde b);\qquad
 (\beta\ne0)\quad h=-\frac{2}{\beta}\log(a+T(b))+\tilde a+T(\tilde b),
\]
where $a,b,\tilde a,\tilde b$ are arbitrary functions of $u_{-1},u,u_1$.
Redefining of $c$ makes possible to set $\tilde a=\tilde b=0$ without loss of
generality, and we arrive to the lattices of types III and IV.

Now, let $\lambda\ne0$. Let us denote $1/\lambda=-a_{-1,1}$, then the solution
of equations (\ref{frR}), (\ref{rRla}) reads
\[
 r=A(u,u_1,v),\quad
 R=a_{-1,1}T(a_{-1,1})\frac{A'}{A^2},\quad
 f=\frac{a_{-1,1}}{T^{-1}(A)A}+b(u_{-1},u,u_1),
\]
where $v=T(a_{-1})-a_1$ and prime denotes the derivative with respect to $v$.
Next, equations (\ref{rR}) amount to $A'=p(u,u_1)A+\alpha$, moreover,
$\beta=-2\alpha$. If $p=0$ then we obtain the type V by choosing $\alpha=1$ and
$A=v$, without loss of generality; it $p\ne0$ then we arrive to the type VI.
\end{proof}

\begin{remark}
The presented partition is not disjoint. More precisely, it follows from the
proof that the types I and II ($l=0$) do not mutually intersect, as well as the
types III--VI ($l=1$). However, there exist the lattices (in particular, the
second order symmetries of the equations from the Yamilov list), such that the
conditions (\ref{2.cond-4})--(\ref{2.cond-6}) are fulfilled for both values
$l=0,1$. These lattices cast simultaneously into two types.
\end{remark}

Further analysis of the integrability conditions requires a separate study of
the above types and it is beyond the scope of this paper. Instead of this, we
will consider a much more simpler classification problem under the additional
assumption of M\"obius invariance.

\section{Classification of M\"obius invariant equations}\label{s:class}

Let us introduce the notation
\[
 X=\frac{(u_1-u)(u_{-1}-u_{-2})}{(u_1-u_{-1})(u-u_{-2})},\quad
 Y=\frac{(u_1-u)(u-u_{-1})}{u_1-u_{-1}}.
\]
The quantities $X$ and $u_{,t}/Y$ are invariants of the group of
linear-fractional transformations $u_j\to\frac{\alpha u_j+\beta}{\gamma
u_j+\delta}$. The general form of the lattice equations (\ref{ut2}) which are
preserved under such substitutions reads
\begin{equation}\label{utYF}
 u_{,t}=YF(X,T(X)).
\end{equation}
The classification problem amounts to determination of the function $F$ from
the conditions (\ref{2.cond-1})--(\ref{2.cond-6}). Comparing to the general
case, here we start from a function of 2 variables instead of 5, which, of
course, is a radical simplification. Further on, we will denote
$F^{(i)}=\partial F(X,T(X))/\partial T^i(X)$.

The reasoning in this section is independent of the proof of Theorem
\ref{th:2.types}, but, actually, it proceeds along the same lines, mutatis
mutandis. In particular, it is convenient, like in Theorem \ref{th:2.types}, to
start the analysis from the condition (\ref{2.cond-4}). It takes the following
form, depending on the exponent $l$:
\begin{align}
\label{Z40}
 (l=0)&\qquad rY^2F^{(1)}=T^2(rY^2F^{(0)}),\\
\label{Z41}
 (l=1)&\qquad (u-u_{-1})^2T^{-1}(r)F^{(1)}=(u_2-u_1)^2T(rF^{(0)}).
\end{align}
The following lemma helps to resolve these equations.

\begin{lemma}\label{l:XX}
Let function $q$ of variables $u_j$ satisfy an equation of the form
\[
 (T^2-1)q=A(X,T(X)),
\]
then $q=\const$, $A=0$.
\end{lemma}
\begin{proof}
It is clear that $q$ may depend, at most, on $u_{-2},u_{-1},u$. The
differentiation with respect to $u_{-2},u_2$ yields $A^{(0,1)}=0$, that is,
$A=a(X)-b(T(X))$. The differentiation with respect to $u_{-2},u_1$ gives
\[
 a'X_{-2}=-q_{-2},\quad a''X_{-2}X_1+a'X_{-2,1}=0.
\]
Taking into account the identity $(\log X)_{-2,1}=0$, the equation $a''X+a'=0$
yields $a=\kappa\log X+\lambda$. In a similar way, $b=\mu\log T(X)+\nu$, and
our equation takes the form
\[
 (T^2-1)(q)=\kappa\log X+\lambda-\mu T(\log X)-\nu.
\]
It is easy to prove, by use of explicit expression of $X$, that this equality
holds only if $\kappa=\mu=0$, $\lambda=\nu$.
\end{proof}

\begin{statement}\label{st:PSL-types}
Up to a constant factor, all solutions of equations (\ref{Z40}) or (\ref{Z41})
are the following:
\begin{alignat}{3}
\label{case0}
 (l=0)&\qquad&& F=g(X+T(X)),&\quad & r=\frac{1}{Y^2g'(X+T(X))},\\
\label{case1}
 (l=1)&&& F=g(X)+g(T(X)),&& r=\frac{1}{(u_1-u)^2},\\
\label{case2}
 (l=1)&&& F=\frac{1}{g(X)g(T(X))}+\delta, &&
          r=\frac{g(T(X))}{(u_1-u)^2},\quad \delta=\const.
\end{alignat}
\end{statement}
\begin{proof}
{\em Case $l=0$.} Let $rY^2F^{(0)}=q$, then equation (\ref{Z40}) takes the form
$T^2(q)F^{(0)}=qF^{(1)}$. According to Lemma \ref{l:XX}, $q=\const$, therefore
$F$ satisfies the equation $F^{(0)}=F^{(1)}$. This brings to solution
(\ref{case0}).

{\em Case $l=1$.} The function $h=(u-u_{-1})^2T^{-1}(r)$ satisfies the equation
\begin{equation}\label{hF}
 hF^{(1)}=T^2(h)T(F^{(0)}),\quad h=h(u_{-2},\dots,u_1).
\end{equation}
Let us prove that this implies $h=h(X)$. Differentiation with respect to
$u_{-2},u_2$ gives
\[
 \frac{h_{-2}}{h}+\frac{F^{(0,1)}}{F^{(1)}}X_{-2}=0,\quad
 (\log F^{(1)})^{(0,1)}=0,
\]
hence $F^{(1)}=a(X)b(T(X))$, $h=p(u_{-1},u,u_1)/a(X)$. In a similar way, we
obtain $F^{(0)}=c(X)d(T(X))$, $h=q(u_{-2},u_{-1},u)/d(X)$. Then $q$ satisfies
the relation
\[
 \frac{T^2(q)}{q}=\frac{a(X)b(T(X))}{d(X)c(T(X))}
\]
and Lemma \ref{l:XX} says that $q=\const$, as required. Now, equation
(\ref{hF}) can be rewritten as the system (cf (\ref{frR}))
\begin{equation}\label{FhH}
 F^{(0)}=\frac{H(X)}{h(T(X))},\quad F^{(1)}=\frac{H(T(X))}{h(X)}.
\end{equation}
The cross-differentiation gives
\[
 -F^{(0,1)}=\frac{H(X)h'(T(X))}{h(T(X))^2}
  =\frac{H(T(X))h'(X)}{h(X)^2} \quad\Rightarrow\quad
 h'=\lambda Hh^2,\quad \lambda=\const.
\]
If $\lambda=0$ then $h'=0$ and solution of (\ref{FhH}) is given by
(\ref{case1}) after some change of notation. If $\lambda\ne0$ then we
substitute $H=h'/(\lambda h^2)$ into (\ref{FhH}) and obtain the solution
(\ref{case2}) by integration.
\end{proof}

Now, the problem is reduced to specification of a function $g$ of one variable.
The use of relations (\ref{Z11})--(\ref{Z21}) makes possible to find it up to
few constant parameters which can be finally fixed by checking the conditions
(\ref{2.cond-1})--(\ref{2.cond-6}). The outline of this rather tedious,
although straightforward computation is given in the proof of the following
theorem.

\begin{theorem}\label{th:list}
Equations (\ref{utYF}) satisfying the necessary integrability conditions
(\ref{2.cond-1})--(\ref{2.cond-6}) are exhausted by the following list (up to a
constant factor in the right hand side):
\begin{align}
\label{SV2}
 u_{,t}&= Y(X+T(X)+c),\quad c=\const,\\
\label{SB}
 u_{,t}&= \frac{Y}{(X-1)(T(X)-1)},\\
\label{SGY}
 u_{,t}&= \frac{4Y(1-X-T(X))}{(2X-1)(2T(X)-1)},\\
\label{dSSK}
 u_{,t}&= \frac{Y(1-X-T(X))}{(X-1)(T(X)-1)},\\
\label{drSSK}
 u_{,t}&= \frac{Y}{(X^{1/2}+\varepsilon)(T(X^{1/2})+\varepsilon)},\quad
 \varepsilon^2=1.
\end{align}
In all cases, conditions (\ref{2.cond-4})--(\ref{2.cond-6}) are fulfilled for
$l=1$, in the case (\ref{SV2}) also for $l=0$.
\end{theorem}

\begin{proof}
Substitution of (\ref{case2}) into equations (\ref{Z12}), (\ref{Z13}) brings to
equation
\[
  Xg'(X)+\alpha g(X)+\beta=0,
\]
where $\alpha,\beta$ are integration constants. Its solutions are of the form
$g=X^k+\gamma$ or $g=\log X+\gamma$, up to a constant factor. Substitution into
(\ref{Z21}), (\ref{Z22}) rejects the second solution and refines the first one,
as well as the constant $\delta$ in (\ref{case2}); it turns out, that three
cases are possible:
\[
 g=X-1,\quad \delta=0;\qquad
 g=1/X-1,\quad \delta=-1;\qquad
 g=x^{1/2}\pm1,\quad \delta=0.
\]
Moreover, a direct check proves that all the rest integrability conditions are
fulfilled and we arrive to the lattice equations (\ref{SB}), (\ref{dSSK}) and
(\ref{drSSK}), respectively.

In the case (\ref{case1}), the relations (\ref{Z12})--(\ref{Z22}) are less
informative and give only the equation $Xg'=\alpha g^2+\beta g+\gamma$ with
arbitrary constant coefficients. Nevertheless, the analysis of condition
(\ref{2.cond-1}) proves that $g$ may be only one of the following:
\[
 g=1/X+c;\qquad g=X+c;\qquad g=1/(2X-1).
\]
The first case is rejected by inspection of the rest conditions; two other
cases pass the test and bring, respectively, to the lattice equations
(\ref{SV2}) and (\ref{SGY}).

Finally, substitution of the case (\ref{case0}) into (\ref{Z11}) yields the
equation $g''=kg'$. If $k=0$ then we come to the lattice (\ref{SV2}) again, the
case $k\ne0$ is rejected by checking (\ref{Z21}).
\end{proof}

\section{Miura type substitutions}\label{s:subst}

All found equations are either known or related with the known ones by
difference substitutions of Miura type. Namely:

--- equation (\ref{SV2}) is the second order symmetry of the
Schwarzian Volterra lattice;

--- equation (\ref{SB}) is the Schwarzian Bogoyavlensky lattice
\cite{Papageorgiou_Nijhoff_1996};

--- equation (\ref{SGY}) is related by B\"acklund transformation with the
Garifullin--Yamilov lattice \cite{Garifullin_Yamilov_2012,
Garifullin_Mikhailov_Yamilov_2014};

--- equations (\ref{dSSK}) and (\ref{drSSK}) are related by Miura type substitutions
with the discrete analog of the Sawada--Kotera equation
\cite{Adler_Postnikov_2011}.

\subsubsection*{\mdseries 1) {\em Symmetry of the Schwarzian Volterra lattice}}
\[
 u_{,t_1}=Y,\quad u_{,t_2}=Y(X+T(X)+c).
\]
The arbitrary parameter $c$ corresponds to addition of the first order
symmetry. If we choose $c=-1$ then the right hand side becomes a product of
linear factors:
\begin{eqnarray}
\label{uSV1}
 u_{,t_1}&=& \frac{(u_1-u)(u-u_{-1})}{u_1-u_{-1}},\\
\label{uSV2}
 u_{,t_2}&=& -\frac{(u_1-u)^2(u-u_{-1})^2(u_2-u_{-2})}
                  {(u_1-u_{-1})^2(u_2-u)(u-u_{-2})}.
\end{eqnarray}
The lattice (\ref{uSV1}) is the well-known Schwarzian version of the Volterra
lattice. The substitution
\[
 v=X=\frac{(u_1-u)(u_{-1}-u_{-2})}{(u_1-u_{-1})(u-u_{-2})}
\]
brings to the Volterra lattice and its symmetry:
\begin{eqnarray*}
 v_{,t_1}&=& v(v_1-v_{-1}),\\
 v_{,t_2}&=& v(v_1(v_2+v_1+v)-v_{-1}(v+v_{-1}+v_{-2}))
            -2v(v_1-v_{-1}).
\end{eqnarray*}

\subsubsection*{\mdseries 2) {\em  Schwarzian Bogoyavlensky lattice}}
\[
 u_{,t}=\frac{Y}{(X-1)(T(X)-1)}
\]
or, in the full form,
\begin{equation}\label{uSB}
 u_{,t}=\frac{(u_2-u)(u_1-u_{-1})(u-u_{-2})}
             {(u_2-u_{-1})(u_1-u_{-2})}.
\end{equation}
This equation is related with the modified Bogoyavlensky lattice
\begin{equation}\label{mB}
 v_{,t}=v(v+1)(v_2v_1-v_{-1}v_{-2})
\end{equation}
by any of the following substitutions:
\[
 v=\frac{X}{1-X}
  =\frac{(u_1-u)(u_{-1}-u_{-2})}{(u_1-u_{-2})(u-u_{-1})},\quad
 v=\frac{u_{-1}-u_1}{u_2-u_{-1}},\quad
 v=\frac{u-u_2}{u_2-u_{-1}}.
\]
This example is known and admits a generalization for the Bogoyavlensky
lattices of any order \cite{Papageorgiou_Nijhoff_1996,Nijhoff_1996}.

\subsubsection*{\mdseries 3) {\em Schwarzian Garifullin--Yamilov lattice}}
\[
 u_{,t}=\frac{4Y(1-X-T(X))}{(2X-1)(2T(X)-1)}
   =-2Y\left(\frac{1}{2X-1}+\frac{1}{2T(X)-1}\right).
\]
The substitution
\[
 w=\frac{1}{2X-1}
  =\frac{(u_1-u_{-1})(u-u_{-2})}
    {(u_1-u)(u_{-1}-u_{-2})-(u_1-u_{-2})(u-u_{-1})}
\]
brings to the lattice
\begin{equation}\label{wGY}
 w_{,t}=(w+1)\left(\frac{w(w_1+1)w_2}{w_1}
   -\frac{w(w_{-1}+1)w_{-2}}{w_{-1}}+w_1-w_{-1}\right).
\end{equation}
On the other hand, the same lattice appears as a result of the substitution
\[
 w=vv_1
\]
from the lattice
\begin{equation}\label{GY}
 v_{,t}=(v_1v+1)(vv_{-1}+1)(v_2-v_{-2})
\end{equation}
which was proven to be integrable in papers \cite{Garifullin_Yamilov_2012,
Garifullin_Mikhailov_Yamilov_2014} (notice also, that it appears under the
scalar reduction $V=(v,1)$ from the vectorial lattice $V_{,t}=\langle
V_1,V\rangle\langle V,V_{-1}\rangle(V_2-V_{-2})$ \cite{Adler_Postnikov_2008}).
A composition of these substitutions defines the B\"acklund transformation
between (\ref{SGY}) and (\ref{GY}).

\subsubsection*{\mdseries 4)
{\em Schwarzian discretization of the Sawada--Kotera equation}}
\[
 u_{,t}=\frac{Y(1-X-T(X))}{(X-1)(T(X)-1)}
  =Y\left(1-\frac{XT(X)}{(X-1)(T(X)-1)}\right).
\]
The right hand side of the lattice is factorizable into linear terms, like in
the case of Bogoyavlensky lattice (\ref{uSB}):
\begin{equation}\label{udSSK}
 u_{,t}=\frac{(u_1-u)(u-u_{-1})(u_2-u_{-2})}
             {(u_2-u_{-1})(u_1-u_{-2})}.
\end{equation}
Integrability is verified by the substitution
\[
 v=\frac{u-u_1}{u_2-u_{-1}}
\]
which brings to the discrete Sawada--Kotera equation
\cite{Hu_Clarkson_Bullough, Adler_Postnikov_2011}:
\begin{equation}\label{dSK}
 v_{,t}=v^2(v_2v_1-v_{-1}v_{-2})-v(v_1-v_{-1}).
\end{equation}
One more modification (cf (\ref{wGY}))
\begin{equation}\label{wdSK}
 w_{,t}=(w+1)\left(\frac{w(w_1+1)^2w_2}{w_1}
   -\frac{w(w_{-1}+1)^2w_{-2}}{w_{-1}}+(2w+1)(w_1-w_{-1})\right)
\end{equation}
is related with (\ref{udSSK}) by substitution
\[
 w=\frac{1}{X-1}
  =-\frac{(u_1-u_{-1})(u_0-u_{-2})}{(u_1-u_{-2})(u_0-u_{-1})}.
\]

\subsubsection*{\mdseries 5) {\em Equation (\ref{drSSK})}}
\[
 u_{,t}= \frac{Y}{(X^{1/2}+\varepsilon)(T(X^{1/2})+\varepsilon)},\quad
   \varepsilon^2=1.
\]
Let $\varepsilon=-1$, for the sake of definiteness, then the substitution
$w=1/(X^{1/2}-1)$ brings to equation (\ref{wdSK}). Therefore, equations
(\ref{drSSK}) and (\ref{dSK}) are related by the B\"acklund transformation.

\section{Some generalizations}\label{s:generalizations}

Returning to the general classification problem for the lattices (\ref{ut2}),
it should be noted that our restriction by the M\"obius invariant case is
rather artificial. Indeed, the obtained examples are not isolated, rather they
are members of more general families of equations which contain arbitrary
parameters. The M\"obius invariant equations are distinguished in these
families only by enlargement of the algebra of classical symmetries, but they
do not differ in terms of higher symmetries.

For instance, equation (\ref{SB}) is a particular case of integrable lattice
equation
\[
 u_{,t}=\frac{(u_2-au)(u_1-au_{-1})(u-au_{-2})}
             {(u_2-bu_{-1})(u_1-bu_{-2})}.
\]
Clearly, the M\"obius invariance is broken here. Nevertheless, the substitution
into the modified Bogoyavlensky lattice survives and takes the form
\[
 v_{,t}=v(bv+a)(v_2v_1-v_{-1}v_{-2}),\quad v=\frac{au_{-1}-u_1}{u_2-bu_{-1}}.
\]
Analogously, equation (\ref{dSSK}) is a particular case of the lattice
\[
 u_{,t}=\frac{(u_1-au)(u-au_{-1})(u_2-a^2u_{-2})}
             {a(u_2-au_{-1})(u_1-au_{-2})}
\]
which is related with (\ref{dSK}) by the substitution
\begin{equation}\label{uv-sub}
 v=\frac{au-u_1}{u_2-au_{-1}}.
\end{equation}
In both examples, the presented substitutions can be viewed as a linear
equation with respect to $u$. In fact, this is the Lax equation for the lattice
equation in the variables $v$, and $a$ serves as the spectral parameter, while
$u$ plays the role of wave function. Thus, consideration of these more general
lattice families is quite natural from the standpoint of the
corresponding spectral problems. Both examples admit generalizations for the
lattices of any order. Notice also, that the substitution (\ref{uv-sub}) can be
represented as a composition of substitutions
\[
 v=\frac{(f+a)f_{-1}}{f_1f_0f_{-1}+a},\quad f=-u_1/u,
\]
where variable $f$ satisfies the modified discrete Sawada--Kotera lattice
\cite{Adler_Postnikov_2011}
\[
 f_{,t}= \frac{f(f+a)}{f_1ff_{-1}+a}\left(
 \frac{f(f_1+a)(f_{-1}+a)(f_2f_1-f_{-1}f_{-2})}
   {(f_2f_1f+a)(ff_{-1}f_{-2}+a)}-f_1+f_{-1}\right).
\]
Analogous generalizations can be constructed also for other equations from the
list.

\subsection*{Acknowledgements}

This work was supported by the RFBR grant \# 16-01-00289a.


\end{document}